\documentclass[11pt]{article}
\usepackage[dvipsnames]{xcolor}
\usepackage{amssymb,amsfonts,amsmath,amsthm,amscd,dsfont}

\usepackage{fullpage}
\usepackage[colorlinks]{hyperref}
\usepackage[nameinlink,capitalize,nosort]{cleveref}

\hypersetup{
  linkcolor=[rgb]{0,0,0.4},
  citecolor=[rgb]{0, 0.4, 0},
  urlcolor=[rgb]{0.6, 0, 0}
}

\newcommand{\F}{\mathbb{F}}
\newcommand{\N}{\mathbb{N}}

\newcommand{\CC}{\mathbb{C}}

\newcommand{\vx}{\mathbf{x}}
\newcommand{\vh}{\mathbf{h}}
\newcommand{\vg}{\mathbf{g}}

\newcommand{\va}{\mathbf{a}}
\newcommand{\vy}{\mathbf{y}}
\newcommand{\vv}{\mathbf{v}}
\newcommand{\vu}{\mathbf{u}}
\newcommand{\vb}{\mathbf{b}}
\newcommand{\vc}{\mathbf{c}}
\newcommand{\vd}{\mathbf{d}}
\newcommand{\vq}{\mathbf{q}}

\newcommand{\Span}[1]{\mathrm{span}\{ #1 \}}
\newcommand{\rank}{\mathrm{rank}}

\newcommand{\adj}{\mathrm{Adj}}

\newtheorem{theorem}{Theorem}[section]
\newtheorem{definition}[theorem]{Definition}

\newtheorem{claim}[theorem]{Claim}
\newtheorem{lemma}[theorem]{Lemma}
\newtheorem{corollary}[theorem]{Corollary}

\crefname{claim}{Claim}{Claims}

\crefname{corollary}{Corollary}{Corollaries}

\title{A Quadratic Lower Bound on Determinantal Complexity}
\author{Mrinal Kumar\thanks{Tata Institute of Fundamental Research, Mumbai, India. Email: \texttt{mrinal@tifr.res.in}.  Research supported by the Department of Atomic Energy, Government of India, under project number RTI400112, and in part by research grants from Google and Premji Invest.}
\and
 Ben Lee Volk\thanks{Efi Arazi School of Computer Science, Reichman University, Israel. Email: \texttt{benleevolk@gmail.com}. The research leading to these results has received funding from the Israel Science Foundation (grant number 843/23).}}

\date{}

\begin{document}

\maketitle

\abstract{
We prove an $\Omega(n^2)$ lower bound on the determinantal complexity of the power sum polynomial $\sum_{i=1}^n x_i^n$ over the field of complex numbers.

A similar result was claimed in a recent paper of Sheshadri \cite{Shesh2026}, via an AI-assisted and AI-written proof. Assuming its correctness, this was the first super-linear lower bound for this fundamental algebraic problem for any explicit polynomial. However, the authors of this note were unable to follow the details and verify the argument in \cite{Shesh2026} in spite of considerable effort on their part. 

The proof we provide here is short, (almost) self-contained and seemingly simpler.

\vspace{2\baselineskip}
\noindent \textbf{AI Use: }ChatGPT Astra was used by the authors on multiple occasions to parse through some of the parts of Sheshadri's proof in \cite{Shesh2026} and the outline of the proof in this note came out of this exercise. The final exposition as well as some of the final details in this note are due to its human authors.  
}

\section{Introduction}

\paragraph{Determinantal Complexity.}
The \emph{determinantal complexity} of a polynomial $F(\vx) \in \CC[\vx]$ is the smallest $m \in \N$ such that there is an $m\times m$ matrix $M(\vx)$, whose entries are polynomials of degree at most one in $\CC[\vx]$ such that $\det(M(\vx))$ equals $F$. 

Early interest in obtaining small determinantal representations can be traced back to P\'{o}lya \cite{Polya13} and Szeg\H{o} \cite{Szego13}.
The importance of this notion, however, was first highlighted by Valiant \cite{Val79} who laid the foundations for algebraic complexity theory.

Valiant \cite{Val79} established the fact that every polynomial has a finite determinantal representation, and observed that such representations can be seen as a natural computational model. This fact, together with appropriate notions of reductions and completeness, played a crucial role in Valiant's project of developing a complexity theory for algebraic computation. The question of proving that the determinantal complexity of the permanent polynomial of a symbolic $n\times n$ matrix (over fields of characteristic not equal to two) is superpolynomially growing in $n$ has served as a guiding light in much of the research in algebraic complexity over the last five decades, as it can be seen as a natural algebraic analog of the P vs.\ NP problem.

While it isn't entirely surprising that this question, or more generally the question of proving a super-polynomial lower bound on the determinantal complexity of \emph{any} explicit polynomial is still open, it is surprising that till very recently, no super-linear lower bound was known for this complexity measure. Mignon and Ressayre \cite{MR04} famously showed a quadratic lower bound on the determinantal complexity of the $n \times n$ permanent polynomial (see also Cai, Chen and Li \cite{CCL10}). These lower bounds, however, while being quadratic in the degree of the polynomial (which serves as the trivial lower bound on determinantal complexity for any polynomial), are proved for a polynomial with $n^2$ variables, so they are in fact not larger than the number of variables of the polynomial. Therefore, in the setting of discussing super-linear or similar such weak lower bounds for this complexity measure we prefer to concentrate on $n$-variate polynomials of degree at most $n$.

The best lower bound on the determinantal complexity of an explicit $n$-variate polynomial known till recently (as a function of the number of variables $n$) was a bound of (roughly) $1.5n$ shown by the authors of this paper for the polynomial $\sum_{i=1}^n x_i^n$ in a prior work (\cite{KV22}).

This state of the art is especially surprising in the light of the fact that we have known super-linear lower bounds on fairly strong algebraic circuit models for some time, including algebraic circuits \cite{Strassen1973b,BS83}, algebraic formulas \cite{Kal85} and algebraic branching programs (ABPs) \cite{K19,CKSV22}. While determinantal complexity and ABP complexity are polynomially equivalent, the polynomial blow-up implies that the lower bound for ABP complexity does not immediately carry over for determinantal complexity (see also \cite{CKV24}).

\paragraph{Sheshadri's recent proof \cite{Shesh2026}.}

In one of the relatively early instances of the use of AI in computational complexity, Karthik Sheshadri claimed in an AI-assisted proof \cite{Shesh2026} that the determinantal complexity of the polynomial $\sum_{i=1}^n x_i^n$ is $\Omega(n^2)$. However, at least the authors of this paper were unable to follow the argument, or check the correctness of the claims in spite of substantial effort on their part. The writing of the paper appeared to be heavily AI-based, with substantial use of either advanced algebraic-geometric notions, or non-standard technical jargon in many places, making the result inaccessible to an algebraic complexity theory audience. 

Assuming the correctness of the argument, this proof provided the first super-linear lower bound on the determinantal complexity of any explicit $n$-variate polynomial. Given the significance of such a result, and the substantial interest in understanding the ideas in the proof within the algebraic complexity research community, there appeared to be merit in obtaining a simpler proof (or a simpler exposition) of this result. 

The purpose of this note is to  provide such a proof. The proof presented here indeed happens to be short, simple, to a large extent intuitive, and self-contained apart from the use of a slightly uncommon (and yet, classical) form of Bézout's theorem. While we do not really follow the ideas in Sheshadri's proof enough to say this with any degree of confidence, our impression is that the proof in this note is quite close to (or even exactly the same as) the intended proof in \cite{Shesh2026}. It is quite conceivable that the simplicity of the description of the proof in this note is an artifact of us looking only for asymptotic bounds, as opposed to the proof in \cite{Shesh2026} where attention is also paid to the constants obtained in the bound. We urge the curious reader to look at the proof in \cite{Shesh2026} to form their opinion on this.  

We start by stating the theorem below, discuss some intuition behind the main ideas in the proof in \cref{sec:overview}, followed by the formal details in \cref{sec:formal-details}. In the interest of keeping this note short, we also refer the reader to \cite{KV22} for a more detailed discussion on the prior work on determinantal complexity. 
\begin{theorem}\label[theorem]{thm:main}
For all sufficiently large $n \in \N$, the following is true. 

Let $M_0, M_1, \ldots, M_n \in \CC^{m\times m}$ be matrices and let $M(\vx)$ be defined as $M(\vx) := M_0 + \sum_{i=1}^n x_iM_i$. If $\det(M(\vx)) = \sum_{i=1}^n x_i^n$, then $m = \Omega(n^2)$. 
\end{theorem}

\paragraph*{Acknowledgements.} We are thankful to Abhranil Chatterjee, Ramprasad Saptharishi and Shubhangi Saraf for many insightful discussions on determinantal complexity over the years. Many thanks to Amir Shpilka for his sanguine advice on why it was important to write and share this proof in this rapidly evolving and utterly confusing world of theory research in the age of AI.

\section{Intuition and overview} \label{sec:overview}
The proof proceeds by constructing a system of polynomial equations with finitely many solutions, and bounding the size of this set from above and below. We first prove a lower bound on the size of this set using the properties of the polynomial $F(\vx) := \sum_{i=1}^n x_i^n$, and then prove an upper bound on the size of this set using the properties of the determinantal representation and Bézout's theorem. Comparing these bounds immediately gives us an $\Omega(n\log n)$ lower bound on the size of the representation. To extend this bound to $\Omega(n^2)$, we replace the use of the vanilla Bézout's theorem by a stronger version that uses some structural properties of the polynomials in the polynomial system.

\paragraph*{Strassen's lower bound: }
At a high level, the proof of \cref{thm:main} is motivated by and is similar to Strassen \cite{Strassen73} and Baur-Strassen's \cite{BS83} proof of super-linear lower bounds for general algebraic circuits. The idea there is to consider the system of polynomial equations given by
\[
\frac{\partial F}{\partial x_i}(\vx) = 1,  \quad i \in [n]
\]
and count the number of solutions to this system. From the structure of the polynomial, the set of solutions to this system is finite and is in fact of size at least $\exp\left(\Omega(n\log n)\right)$. The key conceptual idea is to show that if $F$ is computable by an algebraic circuit of size $s$, then this system can have at most $\exp(O(s))$ solutions. Comparing these bounds gives us an $\Omega(n\log n)$ bound on $s$. 

The upper bound on the set of solutions is obtained by showing that the size of this set is equal to the size of the set of solutions of another system of polynomial equations (that we refer to as the \emph{proxy system}, in contrast to the \emph{original system}), this time on $(n+s)$ variables and $O(s)$ equations, where every equation has degree at most $2$. The classical Bézout theorem implies an $\exp(O(s))$ upper bound on the size of the solution set. The crucial facts used here were 
\begin{itemize}
    \item the new system had polynomials of degree at most a constant independent of $n$,
    \item the number of equations was at most $O(s)$,
    \item and the number of solutions was finite and equal to the number of solutions in the original system.
\end{itemize}

\paragraph*{Working with determinantal representation: }The high level strategy in the proof of \cref{thm:main} is to understand whether the above approach can be adapted when we have a small determinantal representation for $F$, instead of a small circuit.

A natural first attempt is to start with the same polynomial system as in Strassen's proof, where we set all the first-order partial derivatives of $\F(\vx)$ to $1$. The key technical issue in analyzing this system is that we do not know of a way of upper bounding the size of the solution set using a proxy system based on the determinantal representation. However, it turns out that a related system of equations, again based on the first-order partial derivatives of $F$, is easier to work with. The intuition that directs us towards this system is incidentally an elementary linear algebraic fact that follows from Jacobi's formula.  Let $M$ be an $m\times m$ matrix whose entries are linear forms in $\vx$ such that $F(\vx) = \det(M(\vx))$. Furthermore, let $M_0, M_1, \dots, M_n \in \CC^{m\times m}$ be constant matrices such that $M = M_0 + \sum_i x_iM_i$. Then, the following lemma is true. 

\begin{lemma}\label[lemma]{lem:intuition-1}
Let $\va \in \CC^n$ be such that  rank of $M(\va)$ equals $(m-1)$. Let $\vu, \vv$ be arbitrary non-zero vectors in $\CC^m$ such that $\vu^T$ is in the left kernel of $M(\va)$ and $\vv$ is in the right kernel of $M(\va)$. 

Then, there exists a non-zero constant $\lambda \in \CC$ such that for every $i \in [n]$, 
\[
\frac{\partial F}{\partial x_i} (\va) = \lambda \vu^T M_i \vv . 
\]
\end{lemma}
The key takeaway from \cref{lem:intuition-1} is that at points $\va \in \CC^n$ where the rank of $M(\va)$ equals $(m-1)$, the first order partial derivatives of $F$ do indeed behave like a degree $2$ polynomial in a new set of auxiliary variables $\vu, \vv$ coming from the left and right kernel spaces of $M(\va)$. Note that we have $2m$ such variables, corresponding to the coordinates of $\vu, \vv$, and thus the size of the determinantal representation makes an appearance. So, we modify the original polynomial system involving setting all first order derivatives to $1$, in order to ensure that we are dealing with $\va$ such that $M(\va)$ has rank exactly $(m-1)$, and then consider the $2m$ auxiliary variables $\vu, \vv$ and get a proxy system involving them. 

The original polynomial system that we work with is given by the following definition.

\begin{definition}\label[definition]{defn:system-of-eqn}
Let $\F$ be the $n$-variate polynomial $\sum_{i=1}^n x_i^n$. Let $\mathcal{S} \subseteq \CC^n$ be the set of solutions to the system of equations given by 
\begin{enumerate}
    \item $F(\vx) = 0$
    \item $x_{n-1} = 1$
    \item For every $i \leq (n-2)$, $\frac{\partial F}{\partial x_i} = \frac{\partial F}{\partial x_{n-1}}$
\end{enumerate}
\end{definition}
This system clearly has a finite number of solutions, and this number is $\exp(\Theta(n\log n))$. 
We observe that the first two conditions ensure that any solution $\va$ to this system satisfies $M(\va)$ has rank \emph{equal} to $(m-1)$. This is shown in \cref{clm:rank-of-M}. This implies the existence of non-zero vectors in the left and the right kernels (which are one-dimensional spaces) of $M(\va)$. Using these vectors, \cref{lem:intuition-1} tells us that constraints on the first order derivatives of $F$ in the system can be written as degree-$2$ constraints in the coordinates of $\vu, \vv$. The proxy system that we construct in variables $\vx, \vu, \vv$ captures this in a natural way. We need to impose a few more constraints in this system to ensure that the number of solutions of the proxy system is finite. The final description of this system is in \cref{defn:proxy-system-eqn}. 

This system has $O(n+m)$ polynomial constraints in it, each of which has constant degree. Thus, by Bézout's theorem, this system has at most $\exp(O(n+m))$ solutions. Comparing this to the lower bound of $\exp(\Theta(n\log n))$ on the number of solutions of the original system in \cref{defn:system-of-eqn} immediately gives us a lower bound of $\Omega(n\log n)$ on $m$. We then observe that this bound can be improved further for the proxy-system we have since the equations in the system happen to be (essentially) set multilinear with respect to the variables $\vx, \vu, \vv$. This allows us to use a version of Bézout's theorem for such systems (\cref{thm:mulhom-bezout}) that gives a better bound on the number of solutions, leading to the proof of \cref{thm:main}. 

We provide the details in the next section. 

\section{The formal proof}
\label{sec:formal-details}


\subsection{A lower bound from the polynomial} 
Recall that we denote by $\mathcal{S}$ the set of solutions to the system of polynomial equations in \cref{defn:system-of-eqn}.
We first observe that for all $n\geq 2$, the set $\mathcal{S}$ is finite, and compute a lower bound on its size.

\begin{claim}
\label[claim]{cl:S-finite}
The set of solutions $\mathcal{S}$ to the polynomial system in \cref{defn:system-of-eqn} is finite, and $|\mathcal{S}| \ge (n-1)^{n-2}$.
\end{claim}

\begin{proof}
From the definition of $F$, we get that $\frac{\partial F}{\partial x_i} = n\cdot x_i^{n-1} $. So, the system of equations stated in \cref{defn:system-of-eqn} is just 
\[
F(\vx) = 0, \quad x_{n-1} = 1
\]
and 
\[
x_i^{n-1} = 1, \quad i \in [n-2]
\]
The $(n-2)$ equations based on the derivatives clearly give us $(n-1)^{n-2}$ distinct settings of values of $x_1, x_2, \dots, x_{n-2}$. The variable $x_{n-1}$ is forced to be one, and for each of these settings of the first $(n-1)$ variables,  the variable $x_{n}$ has at least one and at most $n$ possible assignments that satisfy the first equation $F(\vx) = 0$. Thus, we get that the set $\mathcal{S}$ is finite and
\[
|\mathcal{S}| \geq (n-1)^{n-2}. \qedhere
\]
\end{proof}

\subsection{An upper bound from the determinantal representation}
We now use the $m\times m$ determinantal representation $M(\vx)$ of $F$ to prove an upper bound on the size of $\mathcal{S}$. To this end, we first construct an alternative system of polynomial equations based on $M(\vx)$ such that the number of solutions of the new system is an upper bound on the size of $\mathcal{S}$. We start by making some simple linear algebraic claims. 

\begin{claim}\label[claim]{clm:adjugate-structure}
Let $\va \in \CC^n$ be such that $\rank(M(\va))=m-1$.

For any non-zero vectors $\vu = \vu(\va), \vv = \vv(\va)$ in $\CC^m$ such that $\vu^{T}$ is in the left kernel of $M(\va)$, $\vv$ is in the right kernel of $M(\va)$, there is a non-zero constant $\lambda$ such that 
    \[
    \adj(M(\va)) = \lambda \vv\vu^{T}  .  
    \]
\end{claim}
\begin{proof}
By the assumption, any non-zero $\vu$ in the left kernel is a basis for the kernel, and any non-zero $\vv$ in the right kernel is a basis for the right kernel. 

By the definition of adjugate of a matrix, we get that
\[
\adj(M(\va))\cdot M(\va) = M(\va) \cdot \adj(M(\va)) = \det(M(\va))\cdot I_m, 
\]
where $I_m$ is the $m\times m$ identity matrix. Since $\det(M(\va)) = 0$, every column of $\adj(M(\va))$ is in the right kernel of $M(\va)$. Similarly, every row is in the left kernel. By the assumption, the left and right kernels of $M(\va)$ are one-dimensional spaces. We therefore get that the rows of $\adj(M(\va))$ are in the span of $\vu$ and the columns are in the span of $\vv$. Thus, there exists a non-zero constant $\lambda$ such that $\adj(M(\va)) = \lambda \vv \vu^{T}$. Here, the non-zeroness of $\lambda$ follows from the fact that the rank of $M(\va)$ equals $(m-1)$, and hence, $\adj(M(\va))$ is non-zero. 
\end{proof}

The following classical claim is known as Jacobi's formula. For completeness, we provide a direct proof.

\begin{claim}[Jacobi's formula]\label[claim]{clm:derivatives-from-adjugate}
For every $i$, 
\[
\frac{\partial F(\vx)}{\partial x_i} = \mathrm{Trace}(\adj(M(\vx)) \cdot M_i). 
\]
\end{claim}
\begin{proof}
Since $F(\vx) = \det(M(\vx))$, expanding the determinant in terms of the monomials, and differentiating with respect to $x_i$, we get that 
\[
\frac{\partial F}{\partial x_i} = \sum_{\sigma \in S_{m}} (-1)^{\mathrm{sign}(\sigma)}\left(\sum_{j=1}^m \frac{\partial (M(\vx)_{j, \sigma(j)})}{\partial x_i} \prod_{\ell \neq j} M(\vx)_{\ell, \sigma(\ell)}\right) . 
\]
Here $S_m$ is the set of permutations on $m$ elements. 
Rearranging the summation, we get 
\[
\frac{\partial F}{\partial x_i} = \sum_{j, j'} \frac{\partial (M(\vx)_{j, j'}}{\partial x_i} \left( \sum_{\sigma \in S_{m}: \sigma(j) = j'} (-1)^{\mathrm{sign}(\sigma)}  \prod_{\ell \neq j} M(\vx)_{\ell, \sigma(\ell)}\right) . 
\]
Now, the internal summation within the brackets is really just equal to $\adj(M(\vx))_{j',j}$, and the partial derivative $\frac{\partial (M(\vx)_{j, j'}}{\partial x_i}$ equals $(M_i)_{j,j'}$ since the entries of $M$ have degree at most one in $\vx$. Thus, we get 
\[
\frac{\partial F}{\partial x_i} = \sum_{j,j'}(M_i)_{j,j'}\cdot \adj(M(\vx))_{j',j} , 
\]
which equals the trace of the matrix $(\adj(M)M_i)$ as claimed. 
\end{proof}
An immediate consequence of \cref{clm:adjugate-structure,clm:derivatives-from-adjugate} is the following useful corollary. 
\begin{corollary}\label[corollary]{cor:final-cor}
For every $\va$ such that $\rank(M(\va))=m-1$, and non-zero vectors $\vu, \vv$ such that $\vu^{T}$ is in the left kernel of $M(\va)$, $\vv$ is in the right kernel of $M(\va)$, the following is true: there exists a non-zero constant $\lambda$ such that  for every $i \in [n]$,
\[
\frac{\partial F}{\partial x_i}(\va) = \lambda \vu^T M_i \vv .
\]
\end{corollary}

\begin{proof}
We have that
\[
\frac{\partial F}{\partial x_i}(\va) = \mathrm{Trace}(\adj(M(\va)) \cdot M_i) = \mathrm{Trace}( \lambda \vv \vu^T \cdot M_i).
\]
The first equality is by \cref{clm:derivatives-from-adjugate}, and the second by \cref{clm:adjugate-structure}. Further
\[
(\vv \vu^T \cdot M_i)_{j,j} = \sum_{k=1}^m (\vv \vu^T)_{j,k} (M_i)_{k,j} = \sum_{k=1}^m \vv_j \vu_k (M_i)_{k,j} 
\]
implying that
\[
\mathrm{Trace}( \lambda \vv \vu^T \cdot M_i) = \lambda \sum_{j=1}^m (\vv \vu^T \cdot M_i)_{j,j} = \lambda \sum_{j=1}^m \sum_{k=1}^m \vv_j \vu_k (M_i)_{k,j}  = \lambda \vu^T M_i \vv. \qedhere
\]
\end{proof}

We now give a simple condition that ensures that the rank of $M(\va)$ is $m-1$.

\begin{claim}\label[claim]{clm:rank-of-M}
Let $\va \in \CC^n$ be such that $F(\va) $ equals zero, but at least one of the first order partial derivatives of $F$ is non-zero at $\va$. Then, the rank of the matrix $M(\va)$ \emph{equals} $(m-1)$. 

In particular, for every $\va \in \mathcal{S}$, $\rank(M(\va))=m-1$.
\end{claim}
\begin{proof}
Since $F(\va) = 0$, we get that $F(\va) = \det(M(\va)) = 0$. So, the rank of $M(\va)$ is at most $(m-1)$. We need to argue that it is exactly $(m-1)$. 

Suppose the rank of $M(\va)$ is at most $(m-2)$. In this case, by the chain rule, every first order partial derivative of $F$ can be written as a linear combination of the $(m-1)\times (m-1)$ determinantal minors of $M(\vx)$, which would all be zero when evaluated at $\va$, which would be a contradiction. 
Thus, $\rank(M(\va))=m-1$.

Finally, if $a \in \mathcal{S}$, $F(\va) = 0$ by definition (recall \cref{defn:system-of-eqn}), and $\va_{n-1} = 1$ implies $\frac{\partial F}{\partial x_{n-1}}(\va) \neq 0$.
\end{proof}

The following claim will enable us to ensure uniqueness of solutions to the proxy polynomial system.

\begin{claim}\label[claim]{clm:uniqueness-constraints-u-v}
Let $U = \cup_{\va \in \mathcal{S}} \{\vu : \vu \neq \mathbf{0}, \vu^TM(\va) = \mathbf{0}\}$ and $V = \cup_{\va \in \mathcal{S}} \{\vv : \vv \neq \mathbf{0}, M(\va)\vv = \mathbf{0}\}$. 

Then, there exist vectors $\vg, \vh \in \CC^m$ such that for every $\vu \in U$ and $\vv \in V$, $\vg^T \vu \neq 0$ and $\vh^T \vv \neq 0$. 
\end{claim}
\begin{proof}
By \cref{clm:rank-of-M}, for every $\va \in \mathcal{S}$, the set $U_{\va} := \{\vu : \vu^T M(\va) = \mathbf{0}\}$ is a one-dimensional subspace spanned by a non-zero vector $\vu_a$. Therefore, if $\vg$ is picked such that  $ \vg^T \vu_a \neq 0$ then $ \vg^T \vu \neq 0$ for every non-zero $\vu \in U_{\va}$.
    
    Since $\mathcal{S}$ is a finite set, we can pick $\vg$ to be a non-root of the non-zero polynomial $P(\vy) = \prod_{\va \in \mathcal{S}} \vy^T  \vu_a $ (in variables $\vy$).
    
    The argument for $\vh$ and $V$ is identical.
\end{proof}

Based on these claims, we are now ready to define the new system of equations.

\begin{definition} \label[definition]{defn:proxy-system-eqn}
    This new system of equations $\mathcal{E}$ has three tuples of variables, $\vu = (u_1,\dots, u_m)$, $\vv = (v_1,\dots, v_m)$ and $\vx = (x_1, \dots, x_n)$, and $2m+(n+1)$ polynomial equations. The equations are as follows. 
    \begin{enumerate}
        \item \label{item:left-kernel} The vector $\vu$ is in the left kernel of $M(\vx)$, i.e. $\vu^{T}M(\vx) = 0$. Note that this corresponds to $m$ polynomial equations, each of which has degree at most one in $\vx$ and at most one in $\vu$.
        \item \label{item:right-kernel} The vector $\vv$ is in the right kernel of $M(\vx)$, i.e. $M(\vx)\vv = 0$. This corresponds to $m$ polynomial equations, each of which has degree at most one in $\vx$ and at most one in $\vv$.
        \item \label{item:partial} For every $i \leq (n-2)$, $\vu^{T}M_i\vv = \vu^{T}M_{n-1}\vv$. Thus, there are $(n-2)$ equations of this kind, each being bilinear in $\vu$ and $\vv$. 
        \item \label{item:xn-1}  $x_{n-1} = 1$
        \item \label{item:u-nonzero} $ \vg^T \vu = 1$
        \item \label{item:v-nonzero} $ \vh^T \vv = 1$
    \end{enumerate}
    The vectors $\vg, \vh$ are picked as in \cref{clm:uniqueness-constraints-u-v}.
    
    Let $\mathcal{T} \subseteq \CC^{2m+n}$ denote the set of solutions of the system of equations defined above. 
\end{definition}
We represent points in $\mathcal{T}$ as a three-tuple $(\va,\vb,\vc)$ of vectors, where $\va \in \CC^n$ is the assignment to $\vx$-variables, $\vb$ is the assignment to $\vu$-variables and $\vc$ is the assignment to $\vv$ variables.

The next few claims help us understand the structure of $\mathcal{T}$ better and help us relate it to $\mathcal{S}$. 

\begin{claim}\label[claim]{clm:sols-projections}
If $(\va,\vb,\vc) \in \mathcal{T}$, then $\va$ is in $\mathcal{S}$. 
\end{claim}
\begin{proof}
    By definition, $\vb^{T}, \vc$ are in the left and right kernels of $M(\va)$ respectively. Furthermore, they are both non-zero since they satisfy $\vg^T \vb = 1$, $ \vh^T \vc = 1$. Thus, $M(\va)$ is singular, and $\det(M(\va)) = F(\va) = 0$. 
    
    We also know that $\va_{n-1} = 1$. The partial derivative of $F = \det(M(\vx))$ with respect to the variable $x_{n-1}$ at the input $\va$ is non-zero. This implies from \cref{clm:rank-of-M} that $\rank(M(\va))=m-1$. From \cref{cor:final-cor}, we get that there is a non-zero constant $\lambda$ such that for every $i$, 
    \[
    \frac{\partial F}{\partial x_i} (\va) = \lambda \vb^T M_i \vc . 
    \]
    Thus, the set of equations $(\vu^TM_i\vv = \vu^TM_{n-1}\vv)$ in \cref{defn:proxy-system-eqn} implies that for every $i \leq (n-2)$, $    \frac{\partial F}{\partial x_i}(\va) =     \frac{\partial F}{\partial x_{n-1}}(\va) $. 

    Thus, $\va$ satisfies all the equations in \cref{defn:system-of-eqn}, and hence $\va \in \mathcal{S}$. 
\end{proof}
We claim that every solution $\va \in \mathcal{S}$ ``lifts'' to a unique solution $(\va,\vb,\vc) \in \mathcal{T}$.
\begin{claim}\label[claim]{clm:sols-lift}
    For every $\va \in \mathcal{S}$, there are unique $\vb, \vc$ in $\CC^m$ such that $ (\va, \vb, \vc) \in \mathcal{T}$. 
\end{claim}
\begin{proof}
$\va \in \mathcal{S}$ implies that $a_{n-1}$ must equal $1$, and hence $F$ has at least one non-zero first order partial derivative at $\va$. Thus, $\va$ immediately satisfies the $x_{n-1} = 1$ constraint in \cref{defn:proxy-system-eqn}, and furthermore,  from \cref{clm:rank-of-M}, we get that $M(\va)$ has rank exactly $(m-1)$. Let $\vb', \vc' \in \CC^m$ be non-zero vectors such that $\vb'^T$ spans the left kernel of $M(\va)$ and $\vc'$ spans the right kernel of $M(\va)$. Clearly, setting $\vu, \vv$ to $\vb', \vc'$ respectively will satisfy the kernel constraints (items 1, 2) in \cref{defn:proxy-system-eqn}. Furthermore, from \cref{cor:final-cor}, we get that $\vb', \vc'$ will also satisfy the derivative constraints (item 3) in \cref{defn:proxy-system-eqn}. 

From the definition of $\vg, \vh$ in \cref{clm:uniqueness-constraints-u-v}, we get that $ \vg^T \vb' $ and $ \vh^T \vc'$ are both non-zero. Thus, by scaling $\vb', \vc'$ appropriately, we get non-zero vectors $\vb, \vc$ that satisfy
\[
 \vg^T  \vb =  \vh^T \vc = 1.
\]
Moreover, the pair $(\vb, \vc)$ is unique in the respective kernels of $M(\va)$, and in particular, $\vb, \vc$ are independent of $\vb', \vc'$, and only depend on $\va$.  Furthermore, the constraints involving $\vu, \vv$ in the first three items of \cref{defn:proxy-system-eqn} are all homogeneous in $\vu,\vv$, and since $(\va, \vb', \vc')$ satisfied those as argued above, so does $(\va,\vb,\vc)$. 

Thus, we have obtained a unique ``lift'' of $\va \in \mathcal{S}$ to get a point in $\mathcal{T}$. 
\end{proof}
From \cref{clm:sols-projections,clm:sols-lift} and the finiteness of $\mathcal{S}$, we get the following corollary. 
\begin{corollary}\label[corollary]{cor:T-finite}
    The set $\mathcal{T}$ is finite and 
    \[
    |\mathcal{T}| = |\mathcal{S}|. 
    \]
\end{corollary}

\subsection{A simple $\Omega(n\log n)$ lower bound}
We are now ready to build upon the setup so far to conclude the proof of \cref{thm:main}. As a first step, we notice that the classical Bézout theorem together with the setup above immediately implies a super-linear lower bound on $m$. We recall that such a bound wasn't known before Sheshadri's work, in spite of the considerable attention received by the question of proving a super-linear lower bound on determinantal complexity in the algebraic and geometric complexity research.

From \cref{cor:T-finite}, since $\mathcal{T}$ is finite, we get from  Bézout's theorem that $|\mathcal{T}|$ is at most the product of the degrees of the polynomials in the system of equations, which is at most $2^{2m+(n-1)}$. Comparing this with the lower bound on the size of $\mathcal{S}$, we get that 
\[
2^{2m + (n-1)} \geq |\mathcal{T}| \geq |\mathcal{S}| \geq (n-1)^{(n-2)}. 
\]
This immediately implies $m = \Omega(n\log n)$. 

\subsection{A quadratic lower bound}
We now observe that a slightly stronger bound can be obtained by noting that the system of polynomial equations in \cref{defn:proxy-system-eqn} has additional structure: the variables can be partitioned into three sets, and each polynomial constraint depends on at most two of these sets and has degree at most one in each set. For such systems, a stronger form of Bézout's theorem is known in the literature.

To state it, we introduce some definitions. Let $\vx_1, \ldots, \vx_k$ be $k$ disjoint sets of variables where $|\vx_i| = n_i$. A polynomial $P$ has multi-degree $d_1, \ldots, d_k$ if it has degree at most $d_i$ in $\vx_i$, for every $i \in [k]$.

The multi-homogeneous Bézout's theorem counts the number of solutions of a system of $N$ such polynomials. It is usually attributed to Shafarevich \cite[Section 2.1]{Shafarevich2013}. We also refer to \cite{WikipediaMultihomogeneousBezout} and to \cite[Section 3.2]{JS07}, the latter of which also cites several other references with additional proofs.

These references all assume that the polynomials are multi-homogeneous   and hence count the solutions in the multi-projective variety $\mathbb{P}_{n_1} \times \mathbb{P}_{n_2} \times \cdots \times \mathbb{P}_{n_k}$ (and therefore $|\vx_i| = n_i + 1$).
One can reduce the affine case to the homogeneous case by homogenization. 
These arguments are rather standard in algebraic geometry.

We directly cite the affine version which is the most convenient for us:

\begin{theorem}[Multi-homogeneous Bézout's theorem]
\label[theorem]{thm:mulhom-bezout}
Let $\vx_1, \ldots, \vx_k$ be $k$ disjoint sets of variables such that $|\vx_i| = n_i$. Let $P_1, \ldots, P_N$ be a set of polynomials, where $P_i$ is of multi-degree $d_{i,1}, \ldots, d_{i,k}$, where $d_{i,j}$ is the degree of $P_i$ in the variables $\vx_j$, and $N=n_1 + n_2 + \cdots + n_k$.

For every $i \in [N]$, define the linear form (in variables $y_1, \ldots, y_k)$
\[
\vd_i = d_{i,1} y_1 + d_{i,2}y_2 + \cdots + d_{i,k}y_k.
\]
Then, the number of solutions to the system $\{ P_i (\vx_1, \ldots, \vx_k) = 0 : i \in [N] \}$ is either infinite, or has size at most $B$, where $B$ is the coefficient of the monomial
\[
y_1^{n_1}y_2^{n_2} \cdots y_k^{n_k}
\]
in the polynomial
\[
\vd_1 \vd_2 \cdots \vd_N.
\]
\end{theorem}

The version stated here is from \cite{MS87} (see also \cite{Abelard18}, Theorem 2.45). The bound $B$ is called the \emph{Bézout number} of the system, and the original statement concerns the number of \emph{isolated} solutions, but clearly if the number of solutions is finite, every solution is isolated.

We would like to invoke \cref{thm:mulhom-bezout} for the polynomial system in \cref{defn:proxy-system-eqn}. However, for this theorem to make sense, the number of equations needs to be equal to the number of variables $n_1 + n_2 + \cdots +n_k$ since $\vd_1 \vd_2 \cdots \vd_N$ is a homogeneous polynomial of degree $N$.

Fortunately, we can amend the system in \cref{defn:proxy-system-eqn} slightly to obtain this.

\begin{definition}\label[definition]{def:square-system}
Pick a vector $\vq \in \CC^m$ such that $\vb^T\vq \neq 0$ for every $\vb \in \CC^m$ for which there exists $\va,\vc$ such that $(\va,\vb,\vc) \in \mathcal{T}$. This is possible since $\mathcal{T}$ is finite, by \cref{cor:T-finite}, using an argument identical to that in \cref{clm:uniqueness-constraints-u-v}.

Let $\Gamma \in \CC^{m-1 \times m}$ be a matrix such that $\ker \Gamma = \Span{\vq}$.

Replace the $m$ equations in item \ref{item:right-kernel} of \cref{defn:proxy-system-eqn} by the $m-1$ equations $\Gamma M(\vx) \vv = 0$ in variables $\vx$ and $\vv$.

Finally, add another variable $z$, and add the equation $z \vu^T \vq = 1$.

Let $\mathcal{T'} \subseteq \CC^{2m+n+1}$ denote the set of solutions to this system.
\end{definition}

\begin{claim}\label[claim]{cl:square-system-solutions}
$|\mathcal{T}| = |\mathcal{T'}|$
\end{claim}

\begin{proof}
Recall that we changed the equations in item \ref{item:right-kernel}, and added the equation $z \vu^T q = 1$.

We show a bijection between the two sets.

Every solution of $\mathcal{T}$ corresponds to a unique solution in $\mathcal{T'}$: if $(\va,\vb,\vc) \in \mathcal{T}$ then $M(\va) \vc = 0$ and thus $\Gamma M(\va) \vc = 0$. Further, $\vb^T \vq \neq 0$ by the choice of $\vq$, so there is a unique choice $z=\frac{1}{\vb^T \vq}$ such that $(\va,\vb,\vc,\frac{1}{\vb^T \vq}) \in \mathcal{T'}$.

Conversely, let $(\va,\vb,\vc,\beta) \in \mathcal{T}'$. We will show that $(\va,\vb,\vc) \in T$.

We need to show that $M(\va)\vc = 0$. It holds that $\Gamma M(\va) \vc = 0$, so $M(\va)\vc = \alpha \vq$ for some $\alpha \in \CC$ since it is in $\ker (\Gamma)$. But since $(\va,\vb,\vc, \beta)$ is a solution to the system in \cref{def:square-system}, $\vb$ satisfies the equation $\vb^T M(\va) = 0$ (in item \ref{item:left-kernel} of \cref{defn:proxy-system-eqn}) so
\[
0 = \vb^T M(\va) \vc = \alpha \vb^T\vq.
\]
Since $\beta \vb^T\vq = 1$, $\vb^T\vq \neq 0$, and we must have $\alpha = 0$ and $M(\va)\vc = 0$.
\end{proof}

\begin{lemma}
\label[lemma]{lem:num-solutions-by-mulhom-bezout}
The size of $\mathcal{T}'$, that is, the number of solutions to the system in \cref{def:square-system}, is at most $2^{n-2} \cdot \left( \frac{e(2m-1)}{n-1}\right)^{n-1} $.
\end{lemma}

\begin{proof}
Our sets of variables are denoted $\vx,\vu,\vv,z$, so that, in the notation of \cref{thm:mulhom-bezout}, we have $n_1=n, n_2=n_3=m$ and $n_4 = 1$.
We list the multi-degrees of the equations, in the order listed in \cref{defn:proxy-system-eqn} (but recall we are working with the amended system in \cref{def:square-system}), and deduce the corresponding linear form in the variables $y_1, y_2, y_3,y_4$, in the notation of \cref{thm:mulhom-bezout}.

\begin{enumerate}
    \item $(1,1,0,0)$ ($m$ equations), and for each the corresponding linear form is $y_1+y_2$
    \item $(1,0,1,0)$ ($m-1$ equations, with linear form $y_1 + y_3$)

    Here the number of equations is $m-1$, due to the amendment detailed in \cref{def:square-system}.
    \item $(0,1,1,0)$ ($n-2$ equations, with linear form $y_2+y_3$)
    \item $(1,0,0,0)$ (1 equation, with linear form $y_1$)
    \item $(0,1,0,0)$ (1 equation, with linear form $y_2$)
    \item $(0,0,1,0)$ (1 equation, with linear form $y_3$)
    \item Finally, we added the equation $z \vu^T \vq = 1$, which has multi-degree $(0,1,0,1)$ with linear form $y_2 + y_4$.
\end{enumerate}

The corresponding Bézout number is therefore the coefficient of $y_1^n y_2^m y_3^m y_4$ in the polynomial
\[
(y_1+y_2)^m \cdot (y_1 + y_3)^{m-1} \cdot (y_2+y_3)^{n-2} \cdot y_1 \cdot y_2 \cdot y_3 \cdot (y_2+y_4),
\]
or the coefficient of $y_1^{n-1} y_2^{m-1} y_3^{m-1}$ in
\[
(y_1+y_2)^m \cdot (y_1 + y_3)^{m-1} \cdot (y_2+y_3)^{n-2}.
\]
Expanding the product, this polynomial equals
\[
\left( \sum_{k_1=0}^m \binom{m}{k_1} y_1^{k_1} y_2^{m-k_1} \right) \cdot \left( \sum_{k_2=0}^{m-1} \binom{m-1}{k_2} y_1^{k_2} y_3^{m-1-k_2} \right) \cdot \left( \sum_{k_3=0}^{n-2} \binom{n-2}{k_3} y_2^{k_3} y_3^{n-2-k_3} \right)
\]

Suppose we pick from the first factor a summand $k_1=k$, for some value of $k$, which gives us $y_1^k y_2^{m-k}$. Then in order to contribute to the coefficient of $y_1^{n-1} y_2^{m-1} y_3^{m-1}$ we must pick $k_2 = n-1-k$ from the second factor, and $k_3 = k-1$ from the third. The relevant coefficient is therefore
\[
B := \sum_{k = 0}^m \binom{m}{k} \cdot \binom {m-1}{n-1-k} \cdot \binom{n-2}{k-1}
\]
For simplicity, we use the convention $\binom{i}{j} = 0$ whenever $j > i$ or $j<0$. Since $\binom{n-2}{k-1} \leq 2^{n-2}$, we get
\[
B \leq 2^{n-2} \cdot \sum_{k=0}^n \binom{m}{k} \cdot \binom{m-1}{n-1-k}  = 2^{n-2} \binom{2m-1}{n-1} \le 2^{n-2} \cdot \left( \frac{e(2m-1)}{n-1} \right)^{n-1}. \qedhere
\]
\end{proof}

We can now finish the proof of \cref{thm:main}.

\begin{proof}[Proof of \cref{thm:main}]
Let $M(\vx)$ be a determinantal representation of the polynomial $\sum_{i=1}^n x_i^n$. Let $\mathcal{S}, \mathcal{T}, \mathcal{T}'$ be the sets as defined above.

By \cref{cor:T-finite}, we have that $|\mathcal{T}|$ is finite and $|\mathcal{T}| \ge (n-1)^{(n-2)}$. By \cref{lem:num-solutions-by-mulhom-bezout} and \cref{cl:square-system-solutions}, we have that $|\mathcal{T}| \le 2^{n-2} \cdot \left( \frac{e(2m-1)}{n-1} \right)^{n-1}$.

Combining these bounds, we get that indeed $m=\Omega(n^2)$.
\end{proof}

\bibliographystyle{alphaurlpp}
\bibliography{references}

\end{document}